\documentclass{article}

\usepackage{amssymb}
\usepackage{pifont}
\usepackage{graphicx}
\usepackage{amsmath}

\newtheorem{definition}{Definition}[section]
\newtheorem{theorem}[definition]{Theorem}
\newtheorem{lemma}[definition]{Lemma}

\newtheorem{remark}[definition]{Remark}
\newenvironment{proof}[1][Proof]{{\bf{#1.\ \ }}}

\begin{document}

\title{Convexity of the Exercise Boundary of the American Put Option}
\author{Hsuan-Ku Liu\footnote{hkliu.nccu@gmail.com}\\ Department of Mathematics and Information Education\\ National Taipei University of Education, Taiwan}
\maketitle

\begin{abstract}
This paper studies the parabolic free boundary problem arising from pricing American-style put options on an asset whose index follows a geometric Brownian motion process. The contribution is to propose a condition for that the early exercise boundary is a  convex function. 
\end{abstract}

\section{Introduction}

From a theoretical as well as practical point of view,
the valuation of American-style options has attracted considerable attention in the field of financial mathematics.
Under the Black-Scholes (BS) framework \cite{Black}, Merton \cite{Merton} presented the price of American options in conjunction with an early exercise boundary as a solution to the free boundary problem in the BS equation.
Since that time, considerable effort has been made to solve the free boundary problem associated with the pricing of American options \cite{Barone, Carr, Chen1, Chen2, Chen (2013), Ekstrom, Evans, Gesker, Jacka, Karatzas 1, Kuske}.
Nonetheless, an entirely satisfactory analytic solution has not been found.
Several researchers have concentrated on finding more accurately expansions or simulations for the early exercise boundary, such as\cite{Barles (1995)}, \cite{Carr}, \cite{Chen1}, \cite{Evans}, \cite{Gesker}, \cite{Jacka}, \cite{Kuske}.
An over view of their results indicates that the early exercise boundary of American put options is a convex function when the dividend rate is less than the risk-free rate and that the convexity may break down when the dividend rate exceeds the risk-free rate \cite{Chen (2013)}.
Chen {\it et al.} \cite{Chen2} and Ekstr$\mathrm{\ddot{o}}$m \cite{Ekstrom} proposed a rigorous verification of the supposition that the early exercise boundary is convex when a stock does not pay dividends.
Chen {\it et al.} \cite{Chen (2013)} demonstrated a proof for that the early exercise boundary is not convex when the dividend rate exceeds the risk-free rate. 
Currently, the convexity of the early exercise boundary remains an open problem when the dividend rate is non-zero \cite{Chen (2013)}.

The contribution of this paper is to examine the convexity of the exercise boundary of the American put option. we show that the early exercise boundary $X_f(T)$ is a strictly decreasing convex function if $q+\frac{\sigma^2}{2}\leq r$.

In summary, the following results have been provided for the convexity of the early exercise boundary of an American put option. 
\begin{enumerate}
\item[(a)] The early exercise boundary is convex when $q=0$ \cite{Chen2, Ekstrom}.
\item[(b)] The early exercise boundary is not convex when $r<q$ \cite{Chen (2013)}.
\item[(c)] We show that the early exercise boundary is convex when $q+\frac{\sigma^2}{2} \leq r$.
\end{enumerate}
Therefore, the convexity of the early exercise boundary remains an open problem when $0<q<r<q+\frac{\sigma^2}{2}$.

This paper is organized as follows. In Section 2, we demonstrate properties of the solution $u(s,t)$ as well as the early exercise boundary $s(t)$. In Section 3, we present a proof of the convexity for the early exercise boundary.

\section{Problem statement}
Let $S_T$ denote the stock price at time $T$. We assume that the stock price satisfies the geometric Brownian motion.
A standard argument explains that the expectation 
$$
P(S,T)=\mathbb{E}_x[e^{-r(T_F-T)}\max\{0,K-S_{T_F}\}]
$$ 
solves a parabolic equation, where $r>0$ is the interest rate, $T_F$ is the expiration date and $\psi(S)=\max\{0,K-S\}$ is the payoff function of a put option. The parabolic equation is expressed as the form:
\begin{equation} \label{eq: Lu=0}
\mathcal{L}_{BS}P=0,
\end{equation}  
with the terminal condition $P(S,T_F)=\max\{0,K-S\}$, where  the Black-Scholes operator $\mathcal{L}_{BS}$ is defined as
$$
\mathcal{L}_{S} \equiv \frac{1}{2} \sigma^2 S^{2}\frac{\partial^2}{\partial S^2}+(r-q)S\frac{\partial}{\partial S}-r+\frac{\partial}{\partial T}.
$$
The solution of (\ref{eq: Lu=0}) provides a formula for valuing a European put option. 
For the American counterpart, the price satisfies the following optimal stopping problem
$$
P(S,T)=\mathrm{ess\ sup}_{\tau\in \mathcal{T}_{T,T_F}}\mathbb{E}_x\left[e^{-r\tau}\psi(S_{\tau})\right], 
$$
where $\mathcal{T}$ is the set of all stopping times and $\mathcal{T}_{T,T_F}=\{\tau\in \mathcal{T}| \mathbb{P}(\tau\in [T,T_F])=1\}$, $0\leq  T \leq T_F < \infty$. The details of the optimal stopping problem for arbitrary diffusion processes can be found in Dayanik \cite{Dayanik (2003)} and Lamberton \cite{Lamberton (2009)}. The connection between the free boundary and the optimal stopping problem for the diffusion process was discussed by Kotlow \cite{Kotlow} and Lamberton \cite{Lamberton (2009)}.

We examine the following one-dimensional free boundary problem for linear parabolic equations arising from the problem of valuing an American put option.\\
{\bf Problem (BS)}
\begin{align}
 &\mathcal{L}_{BS}P=0  && X_f(T)<S<\infty,\ 0<T<T_F,\\ 
 &P(S,T_F) = \max\{0,K-S\}&& 0\leq S < \infty, \\             
 &P(S,T) > \max\{0,K-S\}&& X_f(T)<S<\infty,\ 0<T<T_F,\\
 &\lim_{S\rightarrow \infty}P(S,T) = 0 && 0< T <T_F, \label{Eq:Condition B2}\\
 &P(X_f(T),T)=K-X_f(T) && 0< T < \infty, \label{Eq:Condition B4}\\
 &\frac{\partial P}{\partial S}(X_f(T),T)=-1 && 0< T < T_F \label{Eq:Condition B5}. 
\end{align}
The far-field condition (\ref{Eq:Condition B2})
states that an American put option becomes worthless when the stock price becomes very large. This is because there is no possibility of exercising the option early.
The condition (\ref{Eq:Condition B4}) states that the American put option should be exercised to maximize the expected income when the price $S$ at time $T$ falls to the value of $X_f(T)$.  The smooth-pasting condition (\ref{Eq:Condition B5}) holds when the hedging ratio remains continuous across the early exercise boundary (see Kwok \cite{Kwok (1998)}).

The following properties for $P(S,T)$ and $X_f(T)$ are known to be valid (see \cite{Myneni} and \cite{Peskir}).
\begin{theorem}\label{THM: properties AP}
Let $\{X_f,P\}$ be a solution of  Problem (BS).
Then
\begin{enumerate}
  \item[(a)] $X_f(T)$ is a strictly increasing function with $X_f(T_F)=\min\{K, \frac{r}{q}K\}$.
  \item[(b)] $P(S,T)$ is a convex decreasing function of the stock price $S$ with $P_S\in [-1,0]$ for $X_f(T)<S<\infty$ and $0<T<T_F$.
  \item[(c)] $P(S,T)$ is a decreasing function of the time $T$ for $X_f(T)<S<\infty$ and $0<T<T_F$.
\end{enumerate}
\end{theorem}

The numerical results demonstrated that the early exercise boundary of the American put option is a convex function when $r>q$ and that the convexity may break down when $r<q$.
Chen {\it et al.} \cite{Chen2} and Ekstr$\mathrm{\ddot{o}}$m \cite{Ekstrom} verified that the early exercise boundary is convex when $q=0$. Recently, Chen {\it et al.} \cite{Chen (2013)} showed that the early exercise boundary is not convex when $r<q$.

In the following, we demonstrate that the early exercise boundary $X_f(T)$ of an American put option is convex if $q+ \frac{\sigma^2}{2}\leq r$.
\begin{theorem}\label{THM:American put}
Suppose that the process of the stock price satisfies the geometric Brownian motion process. The free boundary $X_f(T)$ of an American put option is a convex function when $q+ \frac{\sigma^2}{2}\leq r$.
\end{theorem}
The proof of this theorem is provided in the next section.

\section{A proof for Theorem \ref{THM:American put}}
To verify the convexity of $X_f(T)$, we change the operator $\mathcal{L}_{BS}$ to an operator with constant coefficients by
\begin{equation}\label{eq: change variables}
S=e^x,\ T=T_F-2t/\sigma^2,\ P(S,T)=u(x,t),\ X_f(T)=e^{s(t)}.
\end{equation}
Then Problem (BS) becomes \\
{\bf Problem (P)}
\begin{align}
                     &\mathcal{L}u=0  && s(t)<x<\infty,\ 0<t<\infty,                                                      \label{Eq:PDE}\\
                           &u(x,0) = \max(0, K-e^x),&&  -\infty <x<\infty,  \label{Eq:Condition 3}\\              
                     &u(x,t) > \max(0, K-e^x)&& s(t)<x<\infty,\ 0<t<\infty,                                       \label{Eq:Condition 1}\\
                     &\lim_{x\rightarrow \infty}u(x,t) = 0 && 0< t <\infty,                                     \label{Eq:Condition 2}\\
                     &u(s(t),t)=K-e^{s(\tau)} && 0< t < \infty,                                \label{Eq:Condition 4}\\
                     &\frac{\partial u}{\partial x}(s(t),t)=-e^{s(t)} && 0< t < \infty, \label{Eq:Condition 5}
\end{align}
where $k=\frac{2r}{\sigma^2}$, $h=\frac{2q}{\sigma^2}$ and the operator $\mathcal{L}$ is defined as $\mathcal{L}=\mathcal{L}_0-\frac{\partial}{\partial t}$ and
$$
\mathcal{L}_0=\frac{\partial^2}{\partial x^2}+(k-h-1)\frac{\partial}{\partial x}-k.
$$

Let $\{s,u\}$ be the solution to (P). We introduce two sets:
$$
C=\{(x,t)\in \mathbb{R}^+\times [0,\infty)|u(x,t)>\max(K-e^x,0)\},
$$
$$
S=\{(x,t)\in \mathbb{R}^+\times [0,\infty)|u(x,t)=\max(K-e^x,0)\}.
$$
The set $C$ is called the continuation region and the set $S$ is the early exercise region.
 


\begin{definition}
Given $t\in [0,\infty)$, the $t$-section of $S$ is defined as
\begin{equation}\label{Eq:t section}
S_t =\{x\in \mathbb{R}^+| u(x,t)=\max(K-e^x,0)\}.
\end{equation}
\end{definition}

Clearly, we have
$$S=\bigcup_{t<\infty}\left( S_t\times\{t\}\right)$$
and
\begin{equation} \label{Eq:s(t)=sup}
s(t)=\sup\{x|x\in S_t\}.
\end{equation}

The continuation region is then represented as
\begin{equation}\label{Eq:Continuation region}
C =\{(x,t);s(t)<x<\infty,\ 0< t < \infty\}.
\end{equation}
According to Theorem \ref{THM: properties AP}, we obtain the following properties for the solution of Problem (P) directly.
\begin{theorem}\label{THM:properties}
Let $\{s,u\}$ be a solution of  (P).
Then
\begin{enumerate}
  \item[(a)] $s(t)$ is a strictly decreasing function with $s(0)=\min\{\log K,\log (\frac{k}{h}K)\}$.
  \item[(b)] $u_x(x,t)<0$ for $(x,t) \in C$.
  \item[(c)] $u_x(x,t)>-e^x$ for $(x,t) \in C$.
\end{enumerate}
\end{theorem}

Since $s(t)$ is not convex when $q>r$, we consider the convexity of $s(t)$ for $k\geq h$ (ie. $r\geq q$) and define $d=\log K$.
Since $s(t)$ is a decreasing function with $s(0)=d$ and $w(x,t)=u(x,t)-(K-e^x)$ for $x<d$, we have $w(s(t),t)=0$, $w_x(s(t),t)=0$, $w_{t}(s(t),t)=0$.
$w_{xx}(s(t),t)=Kk-he^{s(t)}>0$ . Differentiating the equality $w_x(s(t),t)=0$ with respect to $t$ yields $w_{xx}s'(t)+w_{xt}=0$. Hence we have $\frac{w_{xt}}{w_{xx}}=-s'(t)$ at $x=s(t)$. Moreover, differentiating the equality $w_{xx}(s(t),t)=Kk-he^{s(t)}$ with respect to $t$ yields $w_{xxx}s'(t)+w_{xxt}=-hs'(t)e^{s(t)}>0$ since $s'(t)<0$.
\begin{remark}
By the interior regular theorem of Friedman \cite{Friedman (1958)}, the derivatives $u_{xt}$, $u_{xxt}$ and $u_{xxx}$ exist and are Holder continuous in $C$.
\end{remark}

Let
\begin{equation} \label{eq: v}
v=
  \begin{cases}
   \frac{w_{xt}}{w_{xx}} & \text{if}\ (x,t)\in C_d, \\
   -s'(t)          & \text{if}\ x=s(t),
  \end{cases}
\end{equation}
which is well-defined on $\bar{C}_d=\{(x,t)\in \mathbb{R}^2|s(t)\leq x\leq d,0<t<\infty\}$.
Applying the differential operator $\mathcal{L}$ to equality $vw_{xx}=w_{xt}$, we determine
that $v$ satisfies the following equation
\begin{equation}{\label{Eq:equation of v}}
v_{xx}+((k-h-1)+2\frac{w_{xxx}}{w_{xx}})v_x+\frac{\mathcal{L} w_{xx}}{w_{xx}} v-v_{t}=0
\end{equation}
on $C_d=\{(x,t)\in \mathbb{R}^2|s(t)<x<d,0<t<\infty\}$. 

Since $u_x<0$, $u>0$, $u_{t}>0$ by (\ref{eq: change variables}) on $C_d$ and $\mathcal{L}w=Kk-he^x$, we have 
$$
w_{xx}=-(k-h-1)u_x+ku+u_{t}+Kk-he^x>0\ \mathrm{on}\ C_d
$$ 
if $k-h-1\geq 0$. Since $w(x,t)=u(x,t)-(K-e^x)$ on $C_d$, we have $w_{xx}=u_{xx}+e^x$. 
Applying the constant coefficients operator $\mathcal{L}$ to $w_{xx}$ yields 
$$
\mathcal{L} w_{xx}=\mathcal{L}(u_{xx}+e^x)=\frac{\partial^2}{\partial x^2}\mathcal{L}u+\mathcal{L} e^{x}=\mathcal{L} e^{x}=-he^x<0.
$$
We also have $w_{xx}(s(t),t)>0$. Therefore, the equation (\ref{Eq:equation of v}) is a parabolic equation with bounded coefficients if $k-h-1\geq 0$.

Friedman \cite{Friedman} defined the lower $\Omega$-neighbothood as follows.
\begin{definition}
An $\Omega$-neighbothood of a point $(x_0,t_0)$ is the intersection of a neighborhood of $(x_0,t_0)$ with $\Omega$. A lower $\Omega$-neighbothood of a point $(x_0,t_0)$ is the intersection of an $\Omega$-neighbothood of $(x_0,t_0)$ with the half space $t\leq t_0$.
\end{definition}

To show the convexity of $X_f(T)$, it suffices to show that $s(t)$ is a convex function.
Now, we provide a proof of the main contribution in this paper.\\

Goodman and Ostrov [17, p.1831] provided the following estimate for the early exercise boundary $b(t)$.
\begin{theorem}\label{THM: asymtotic}
The asymptotic expansion of $b(t)$ as $t\rightarrow 0$ takes the following form 
$$
\begin{array}{ll}
s(t)\sim& -\sqrt{-2t\log(ct)}, \\
s'(t)\sim& \frac{\log(ct)+1}{\sqrt{-2t\log(ct)}}, \\
s''(t)\sim&  \frac{\log^2(ct)+1}{({-2t\log(ct)})^{\frac{3}{2}}}>0,
\end{array}
$$
where $c=4\pi k^2$.
\end{theorem}
This implies that the early exercise boundary $X_f(T)$ is convex near the maturity for the case of $0\leq q \leq r$.

{\bf Proof of Theorem \ref{THM:American put}.}
We have determined that $s(t)$ is a strictly decreasing function.
Suppose that there is a closed interval $I$ such that $s(t)$ is a concave function on the interval $I=[a,b]$. According to the estimate of $X_f(t)$ near the maturity in Theorem \ref{THM: asymtotic}, we known that $s(t)$ is convex near 0 if $0\leq q \leq r$. Thus $a> 0$. 
Suppose that there exists a $t_0\in I$ with $s'(t_0)=m<0$ because $s(t)$ is strictly decreasing and is differentiable almost everywhere.
Then $s'(t)\leq m$ for almost every $t>t_0$ in $I$.

When $s(t)$ is assumed to be a concave function on $I$, we consider the following two lemmas for the 
level curve $\Gamma_{\alpha}=\{(x,t)\in C_d|v(x,t)=\alpha\}$.
\begin{lemma} \label{lemma: extremum}
Let $v$ be a solution of (\ref{Eq:equation of v}). If $s(t)$ is a concave function on an interval $I$, then for any $t_0\in I$ $v$ can not attain an extremum at $(s(t_0),t_0)$ with respect to any lower $\bar{\Omega}$- neighborhood of $(s(t_0),t_0)$.
\end{lemma}
\begin{proof}
Since $s(t)$ is a concave function on the interval $I$, then $s''(t)<0$ on $I$; this implies that $-s'(t)$ is an increasing function on $I$. Since $v(s(t),t)=-s'(t)$, we conclude that $v$ can not attain a minimum at $(s(t_0),t_0)$ with respect to any lower $\bar{\Omega}$- neighborhood of $(s(t_0),t_0)$ on $I$.

Suppose that $v$ attains a maximum at $(s(t_0),t_0)$. Then 
\begin{equation}\label{eq: v_x(s(t_0),t_0)}
v_x(s(t_0),t_0)\leq 0.
\end{equation}
However, at $(s(t_0),t_0)$,
$$
v_x=(\frac{w_{xt}}{w_{xx}})_x=\frac{w_{xxt}w_{xx}-w_{xxx}w_{xt}}{w_{xx}^2}=\frac{w_{xxt}-w_{xxx}v}{w_{xx}}=\frac{-hs'(t)e^{s(t)}}{w_{xx}}>0,
$$
thus contradicting to (\ref{eq: v_x(s(t_0),t_0)}). 
\end{proof}

\begin{lemma}\label{lemma: monotone gamma}
Let $\Gamma_{\alpha}$ be the level curves on which $v=\alpha$.  If $s(t)$ is a concave function on an interval $I$,
then, for each $\alpha$ there exists a $g_{\alpha}(t)$ such that 
$$
\Gamma_{\alpha}=\{(g_{\alpha}(t),t)|v(g_{\alpha}(t),t)=\alpha,\ t>0\}.
$$
\end{lemma}
\begin{proof}
Since $w_{xx}>0$, $\mathcal{L}w_{xx}< 0$ and $v$ satisfy the parabolic equation (\ref{Eq:equation of v}), the $t$-coordinate along $\Gamma_\alpha$ can not be
(i) first decreasing and then increasing and (ii) first increasing and then decreasing. 
For (i), a region would exist in which the parabolic boundary is a part of $\Gamma_\alpha$; consequently
$v\equiv \alpha$ in this region and $v\equiv \alpha$ in $C_d$.
For (ii), there would be a region with parabolic boundary consisting of a part of $\Gamma_\alpha$ and a part of $\{(s(t),t)|0<t\leq t_0\}$ which implies that an extremum exists at $(s(t_0),t_0)$ with respect to the lower $\Omega$-neighborhood of $(s(t_0),t_0)$. Employing Lemma \ref{lemma: extremum}, we have that the extremum can not appear at $v(s(t_0),t_0)$. Therefore, we conclude that the level curve $\Gamma_{\alpha}$ can not first increasing.
\end{proof}
The idea of this proof is similar to Friedman and Jensen \cite{Friedman} (seeing Page 4 of \cite{Friedman} for the details).

Since $I=[a,b]$ and $a>0$, there is a point $t_0\in I$ with $v(s(t_0),t_0)=-s'(t_0)=-m$ such that the line
$$
y(t)=m(t-t_0)+s(t_0), \ t>0
$$
intersects $s(t)$ at $t_2<t_0$ and $t_0$; that is $t_2=\inf\{t|(y(t),t)\in C_d\}$ with $y(t_2)=s(t_2)$ and $y(t_0)=s(t_0)$.
Since $t_0\in I$, we have $v(s(t),t)=-s'(t)\geq -m$ for $t>t_0$ in $I$.
Since $s(t)$ is bounded below and $m<0$, there must exist another point
$t_1>t_0$ such that $y(t_1)=s(t_1)$.
Now, we have $s(t_i)=y(t_i)$, $i=0,1,2$. 

We also have $w_x=u_x+e^x>0$ on $C_d$ according to (c) in Theorem \ref{THM:properties}. Let $f(t)=w_{x}(y(t),t)=u_{x}(y(t),t)+e^{y(t)}>0$ for some $t>t_2$.
Thus, we derive that
\begin{equation} \label{Eq:f'(t)}
\begin{array}{ll}
\displaystyle f'(t)&=\displaystyle mw_{xx}(y(t),t)+w_{xt}(y(t),t)\\
\displaystyle      &=\displaystyle w_{xx}(y(t),t)(m+v(y(t),t))
\end{array}
\end{equation}
for $t>t_2$.
Since $y(t_0)=s(t_0)$ and $v(s(t_0),t_0)=-s'(t_0)=-m$, we obtain
$$
\begin{array}{ll}
f'(t_0)& =w_{xx}(y(t_0),t_0)(m+v(y(t_0),t_0))\\
       & =w_{xx}(s(t_0),t)(m+v(s(t_0),t))=0,
\end{array}
$$
We also have $w_x(s(t),t)=0$ by (\ref{Eq:Condition 5}). Since  $y(t_i)=s(t_i)$, $i=0,1,2$ and $w_x(x,t)>0$ for $(x,t)\in C_d$, we also have $f(t_i)=w_x(y(t_i),t_i)=w_x(s(t_i),t_i)=0$, $i=0,1,2$  and
$(y(t),t)\in C_d$ for $t\in (t_2,t_1)$.
Thus, a local maximum of $f$ exists in $(t_0,t_1)$ and $(t_2,t_0)$, namely $f(t_3)$ and $f(\bar{t}_3)$ where
$t_3\in (t_0,t_1)$ and $\bar{t}_3\in (t_2,t_0)$.
This implies that $f'(t_3)=0$ and $f'(\bar{t}_3)=0$.
Since $w_x=u_x+e^x$ is a solution of parabolic equation and $f(t)=w_x(y(t),t)$, which does not oscillate as $t\rightarrow t_0$.
This implies that $f(t)$ do not produce an infinite sequence
of local maximum, the locations of which tends to $t_0$.
We can therefore assume that $t_3$ and $\bar{t}_3$ are the first maximum from $t_0$ and no local maximum exists between $t_0$ and $t_3$ and between $\bar{t}_3$ and $t_0$.
By the same reason, there also exists a point $\bar{t}_3\in (t_2,t_0)$ such that $f'(\bar{t}_3)=0$. 
Since $f(t_0)=f(t_1)=0$, $f(t_3)>0$, and $f'(t_i)=0$, $i=0,3$,
we have
\begin{equation} \label{Eq:f'(t)<0}
f'(t)>0\ \mathrm{for}\ t\in (t_0,t_3)
\end{equation}
and
\begin{equation} \label{Eq:f'(t)>0}
f'(t)<0\ \mathrm{for}\ t\in (t_3,t_4),
\end{equation}
where $t_3< t_4 \leq  t_1$.

Let $\Gamma_{-m}$ be the level curves on which $v=-m$. According to Lemma \ref{lemma: monotone gamma}, there exists the $g_{-m}(t)$ such that
$$
\Gamma_{-m}=\{(g_{-m}(t),t)|v(g_{-m}(t),t)=-m,t>0\}.
$$
Since $f'(t_i)=0,\ i=0,3$ and $f'(t)=w_{xx}(y(t),t)(m+v(y(t),t))$,
we have $v(y(t_i),t_i)=-m$, $i=0,3$, which implies that
$(y(t_i),t_i)\in \Gamma_{-m}$, $i=0,3$.
Next, we consider the function $g_{-m}(t)$.
Since $(y(t_i),t_i)\in \Gamma_{-m}$, that is 
\begin{equation}\label{eq: v(y(t_i),t_i)=-m}
v(y(t_i),t_i)=-m,\ \ i=0,3,
\end{equation}
we have $y(t_i)=g_{-m}(t_i), \ i=0,3$. Since $f'(t)=w_{xx}(y(t),t)(m+v(y(t),t))>0$ for $t\in (t_0,t_3)$ by (\ref{Eq:f'(t)<0}) and (\ref{Eq:f'(t)}) and $w_{xx}(y(t),t)>0$ by the assumption,
this implies that
\begin{equation}{\label{Eq:v(y(t),t)>-m}}
v(y(t),t)>-m,\ \mathrm{for}\ t\in (t_0,t_3).
\end{equation}

Since $g_{-m}(t)$ is continuous on $(t_2,t_1)$, we  have only
the following two cases:
(1) $y(t)>g_{-m}(t)$ for $t\in (t_0,t_3)$, and
(2) $y(t)<g_{-m}(t)$ for $t\in (t_0,t_3)$.

We first consider case (1).
Since $g_{-m}(t_0)=y(t_0)=s(t_0)$ and $y(t)>g_{-m}(t)>s(t)$ for $t\in (t_0,t_3)$,
there is a $\delta>0$ such that $y'(t)>g'(t)>s'(t)$
for $t\in (t_0,t_0+\delta)$.
Since $y'(t)=m$, we have $v(s(t),t)=-s'(t)>-y'(t)=-m$ for $t\in (t_0,t_0+\delta)$.
Let $\Omega=\{(x,t)|s(t)\leq x \leq y(t),\ t_0\leq t\leq t_0+\delta\}$.
On $\Omega$, we have $t'$, $t''\in (t_0,t_0+\delta)$ such that $v(s(t'),t')=v(y(t''),t'')=\beta>-m$, but $v(g_{-m}(t),t)=-m$ for all $t\in (t_0,t_0+\delta)$.
 This implies that there exists a level curve, say $\Gamma_{\beta}$, crosses $g_{-m}(t)$ connected $s(t')$ and $y(t'')$. This contracts to $\Gamma_{\beta}\cap \Gamma_{-m}\neq \emptyset$, $\beta\neq -m$. Therefore, case (1) does not hold.

Next, we consider case (2). We know that the level curves
$\Gamma_\alpha$ of a parabolic equation are continuous.
Since $f'(\bar{t}_3)=w_{xx}(m+v(y(\bar{t}_3),\bar{t}_3))=0$, we also have $v(y(\bar{t}_3),\bar{t}_3)=-m$; that is $(y(\bar{t}_3),\bar{t}_3)\in \Gamma_{-m}$.
Consider the line $y(t)$ for $t\in (t_2,t_0)\cup (t_0,t_3) $.
In (\ref{Eq:v(y(t),t)>-m}), we have $v(y(t),t)>-m$ for $t\in (t_0,t_3)$.
We also have $f(t_0)=0$ and $f(t)=w_x(y(t),t)>0$ for $t\in (t_2,t_0)$.
This implies that there is a $\delta_2>0$ such that $f'(t)<0$ for $t\in (t_0-\delta_2,t_0)$.
Since $f'(t)=w_{xx}(y(t),t)(m+v(y(t),t))$ and $f'(t)<0$ for $t\in (t_0-\delta_2,t_0)$ and $w_{xx}>0$ for $(x,t)\in C_d$, we obtain
\begin{equation} \label{eq: v(y(t),t)<-m}
v(y(t),t)<-m
\end{equation}
for $t\in (t_0-\delta_2,t_0)$.
Now, we have only the following two subcases for case (2): (2.1) $g_{-m}(t)>y(t)$ for $t\in (t_0-\delta_2,t_0)$ and 
(2.2) $g_{-m}(t)<y(t)$ for $t\in (t_0-\delta_2,t_0)$.

 For case (2.1), we can select a suitable $\delta>0$ such that $v(y(t),t)<-m$ for $t\in (t_0-\delta,t_0)\cup (t_3,t_3+\delta)$, $t_3+\delta<t_4$ by (\ref{Eq:f'(t)>0}). Since $v(y(t_0),t_0)=-m=v(y(t_3),t_3)$ by (\ref{eq: v(y(t_i),t_i)=-m}) and $v(y(t),t)<-m$ for $t\in (t_0-\delta,t_0)\cup (t_3, t_3+\delta)$, there exists a $t'\in (t_0-\delta,t_0)$ and a $t''\in (t_3,t_3+\delta)$
such that
$$
v(y(t'),t')=\beta=v(y(t''),t''),\ \mathrm{for\ some}\ \beta<-m.
$$
Since the level curves of a parabolic equation are continuous, there
exists a level curve $\Gamma_\beta$ connecting
$(y(t'),t')$ and $(y(t''),t'')$.
There is an intersection of $\Gamma_{-m}$ and $\Gamma_{\beta}$ on $(t_0-\delta, t_0)$.
This contradicts to $\Gamma_{-m}\cap \Gamma_{\beta}\neq \emptyset$.

For case (2.2), we have $v(y(t),t)<-m$ for $t\in (t_0-\delta,t_0)$ by (\ref{eq: v(y(t),t)<-m}) and $v(g_{-m}(t),t)=-m$ for $t\in (t_0-\delta,t_0)$. If $v(s(t),t)<-m$ for $t\in (t_0-\delta,t_0)$, there exists a level curve, say $\Gamma_{\alpha}$, crosses over $g_{-m}(t)$ connected $s(t)$ and $y(t)$. This contradicts to $\Gamma_{\alpha}\cap \Gamma_{-m}\neq \emptyset$, $\alpha\neq -m$. 
If $v(s(t),t)>-m$ for $t\in (t_0-\delta,t_0)$, we have $v(y(t),t)>-m$ on $(t_0,t_3)$ by (\ref{Eq:v(y(t),t)>-m}) and $v(y(t_0),t_0)=v(s(t_0),t_0)=-m$. This implies that there exists a $t'\in (t_0-\delta,t_0)$ and a $t''\in (t_0,t_3)$ such that
$$
v(s(t'),t')=\beta=v(y(t''),t''),\ \mathrm{for\ some}\ \beta>-m.
$$
Since the level curves of a parabolic equation are continuous, there
exists a level curve $\Gamma_\beta$ connecting
$(y(t'),t')$ and $(y(t''),t'')$.
This contradicts $\Gamma_{-m}\cap \Gamma_{\beta}\neq \emptyset$.
Terefore, case (2) does not hold.

Both case (1) and case (2) do not hold; therefore we conclude that 
$s(t)$ can not be a concave function in any interval.
Thus, $s(t)$ is a convex function.

\begin{remark}
Given $\alpha\in R$ and $g_{\alpha}(t)$ as the function, such that
$$
v(g_{\alpha}(t),t)=\alpha
$$
with $g_{\alpha}(t_0)=s(t_0)$, where $v(s(t_0),t_0)=\alpha$.
Then 
$$
\frac{dv}{dt}=v_x\frac{dg_{\alpha}(t)}{dt}+v_t=0.
$$
According to Sard's lemma, the set of $v_x(x,t)=0$ is measure zero. 
Thus, $-\frac{v_t}{v_x}$ is defined for almost every point on $\Omega$.
We consider the following IVP
\begin{equation} \label{eq: dg}
\frac{dg_{\alpha}(t)}{dt}=-\frac{v_t}{v_x}\ \ \ (a.e.)
\end{equation}
with $g_{\alpha}(t_0)=s(t_0)$.
Indeed, the  weak solution for (\ref{eq: dg}) exists.
Therefore $g_{\alpha}(t)$ is continuous for all $t$ with $v(g_{\alpha}(t),t)=\alpha$.
\end{remark}


\end{document}